\documentclass{amsart}
\usepackage{amsfonts}
\usepackage{amsmath}
\usepackage{amssymb}

\theoremstyle{plain}
\newtheorem{theorem}{Theorem}[section]
\newtheorem{lemma}[theorem]{Lemma}
\newtheorem{proposition}[theorem]{Proposition}
\theoremstyle{definition}

\theoremstyle{remark}
\newtheorem{remark}[theorem]{Remark}
\numberwithin{equation}{section}

\input{tcilatex}

\begin{document}
\title{The differential-algebraic and bi-Hamiltonian integrability analysis
of the Riemann type hierarchy revisited}
\author{Yarema A. Prykarpatsky$^{1},$ \ Orest D. Artemovych$^{2},$ Maxim V.
Pavlov$^{3}$ and Anatoliy K. Prykarpatsky$^{4}$ }
\address{ $^{1}$The Department of Applied Mathematics at the Agrarian
University of Krakow, Poland, and the Institute of Mathematics of NAS, Kyiv,
Ukraine}
\email{yarpry@gmail.com}
\address{$^{2}$ The Department of Algebra and Topology at the Faculty of
Mathematics and Informatics of the Vasyl Stefanyk Pre-Carpathian National
University, Ivano-Frankivsk, Ukraine, and the Institute of Mathematics and
Informatics at the Tadeusz Kosciuszko University of Technology, Crac\'{o}w,
Poland; }
\email{artemo@usk.pk.edu.pl}
\address{$^{3}$The Department of Mathematical Physics, P.N. Lebedev Physical
Institute of Russian Academy of Sciences, Moscow, 53 Leninskij Prospekt,
Moscow, Russia, and \ Laboratory of Geometric Methods in Mathematical
Physics, Department of Mech. \& Math., Moscow State University, 1 Leninskie
gory, Moscow, Russia }
\email{maxim@math.sinica.edu.tw }
\address{$^{4}$ AGH University of Science and Technology, the Department of
Mining Geodesy, Crac\'{o}w 30059, Poland }
\email{pryk.anat@ua.fm}
\maketitle

\begin{abstract}
A differential-algebraic approach to studying the Lax type integrability of
the generalized Riemann type hydrodynamic hierarchy is revisited, its new
Lax type representation is \ constructed in exact form. The related
bi-Hamiltonian integrability and compatible Poissonian structures of the
generalized Riemann type hierarchy are also discussed by means of the
gradient-holonomic and geometric methods.
\end{abstract}

\section{Introduction}

Recently new mathematical approaches, \ based on differential-algebraic and
differential geometric methods and techniques, were applied in works \cite%
{GPPP,Wi,PAPP} for studying the Lax type integrability of nonlinear
differential equations of Korteweg-de Vries and Riemann type. \ In
particular, a great deal of analytical studies \cite%
{PP,PoP,GBPPP,Pav,GPPP,Wan} were devoted to finding the corresponding
Lax-type representations of the infinite Riemann type hydrodynamical
hierarchy 
\begin{equation}
D_{t}^{N}u=0,\ \ \ \ D_{t}:=\partial /\partial t+u\partial /\partial x,
\label{r01}
\end{equation}%
where $N\in {\mathbb{Z}_{+}},$ $(x,t)^{\intercal }\in \mathbb{R}^{2}$ and $%
u\in C^{\infty }(\mathbb{R}/2\pi \mathbb{Z};\mathbb{R}).$ It was found that
the related dynamical system 
\begin{equation}
D_{t}u_{1}=u_{2},\ ...,D_{t}u_{j}=u_{j+1},...,D_{t}u_{N}=0,\ \ \ \ 
\label{r02}
\end{equation}%
defined on a $2\pi $-periodic infinite-dimensional smooth functional
manifold $M_{N}\subset C^{\infty }(\mathbb{R}/2\pi \mathbb{Z};\mathbb{R}%
^{N}),$ possesses \cite{PoP,PAPP} \ for an arbitrary integer $N\in \mathbb{Z}%
_{+}$ a suitable Lax type representation 
\begin{equation}
D_{x}f=l_{N}[u;\lambda ]f,\ \ \ D_{t}f=q_{N}(\lambda )f  \label{r03}
\end{equation}%
with $\lambda \in \mathbb{C}$ being a complex spectral parameter and $f\in
L_{\infty }(\mathbb{R};\mathbb{C}^{N})$ and \ matrices $l_{N}[u;\lambda
],q_{N}(\lambda )\in End\mathbb{C}^{2}.$ Here, by definition, $u_{1}:=u\in
C^{\infty }(\mathbb{R}^{2};\mathbb{R})\ $and\ the differentiations \ 
\begin{equation}
D_{t}:=\partial /\partial t+u_{1}D_{x},\ \ D_{x}:=\partial /\partial x
\label{r03a}
\end{equation}%
satisfy \ on the manifold $M_{N}$ \ the following commutation relationship: 
\begin{equation}
\lbrack D_{x},D_{t}]=(D_{x}u_{1})D_{x}.  \label{r04}
\end{equation}%
In particular, for the cases $N=\overline{2,3}$ and $N=4$ the following
exact matrix expressions 
\begin{eqnarray}
&&l_{2}[u;\lambda ]=\left( 
\begin{array}{cc}
\lambda u_{1,x} & u_{2,x} \\ 
-2\lambda ^{2} & -\lambda u_{1,x}%
\end{array}%
\right) ,\ \ \ q_{2}(\lambda )=\left( 
\begin{array}{cc}
0 & 0 \\ 
-\lambda & 0%
\end{array}%
\right) ,  \notag \\
&&  \label{r05} \\
&&l_{3}[u;\lambda ]=\left( 
\begin{array}{ccc}
\lambda ^{2}u_{1,x} & -\lambda u_{2,x} & u_{3,x} \\ 
3\lambda ^{3} & -2\lambda ^{2}u_{1,x} & \lambda u_{3,x} \\ 
6\lambda ^{4}r_{3}^{(1)}[u] & -3\lambda ^{3} & \lambda ^{2}u_{1,x}%
\end{array}%
\right) ,\text{ \ }q(\lambda ):=\left( 
\begin{array}{ccc}
0 & 0 & 0 \\ 
\lambda & 0 & 0 \\ 
0 & \lambda & 0%
\end{array}%
\right) ,  \notag \\
&&  \notag \\
&&l_{4}[u;\lambda ]=\left( 
\begin{array}{cccc}
-\lambda ^{3}u_{1,x} & \lambda ^{2}u_{2,x} & -\lambda u_{3,x} & u_{4,x} \\ 
-4\lambda ^{4} & 3\lambda ^{3}u_{1,x} & -2\lambda ^{2}u_{2,x} & \lambda
u_{3,x} \\ 
-10\lambda ^{5}r_{N}^{(1)}[u] & 6\lambda ^{4} & -3\lambda ^{3}u_{1,x} & 
\lambda ^{2}u_{2,x} \\ 
-20\lambda ^{6}r_{4}^{(2)}[u] & 10\lambda ^{5}r_{4}^{(1)}[u] & -4\lambda ^{4}
& \lambda ^{3}u_{1,x}%
\end{array}%
\right) ,\text{ \ }q_{4}(\lambda ):=\left( 
\begin{array}{cccc}
0 & 0 & 0 & 0 \\ 
\lambda & 0 & 0 & 0 \\ 
0 & \lambda & 0 & 0 \\ 
0 & 0 & \lambda & 0%
\end{array}%
\right) ,\ \ \   \notag
\end{eqnarray}%
polynomial in $\lambda \in \mathbb{C},$ \ were presented in exact form. \ A
similar Lax type representation, as it follows from the statements of the
work \cite{PAPP}, also holds for the general case $N\in \mathbb{Z}_{+}.$
These results are mainly based on a new differential-algebraic approach,
devised recently in work \cite{PAPP} and contain complicated enough \cite%
{PoP} functional expressions $r_{N}^{(j)},$ $j=\overline{1,N-2},$ which
satisfy the recurrent differential-functional equations

\begin{equation}
\left\{ 
\begin{array}{l}
D_{t}r_{N}^{(1)}+r_{N}^{(1)}D_{x}u^{(1)}=1, \\ 
D_{t}r_{N}^{(2)}+r_{N}^{(2)}D_{x}u^{(1)}=r_{N}^{(1)}, \\ 
... \\ 
D_{t}r_{N}^{(j+1)}+r_{N}^{(j+1)}D_{x}u^{(1)}=r_{N}^{{(j)}},\ \ \ j=\overline{%
1,N-3}, \\ 
... \\ 
D_{t}r_{N}^{(N-2)}+r_{N}^{(N-2)}D_{x}u^{(1)}=r_{N}^{{(N-3)}}%
\end{array}%
\right.  \label{r06}
\end{equation}%
on the functional manifold $M_{N}.$ $\ $Recently Popowicz \cite{Po} has
found a new Lax representation for the generalized Riemann equation \ (\ref%
{r02}) at arbitrary $N\in \mathbb{Z}_{+}.$ In the present article, being
strongly based on the methods devised in works \cite{PAPP,PoP}, we will
revisit the differential-algebraic approach to the Riemann type
hydrodynamical hierarchy \ (\ref{r01}) and construct its new Lax type
representation 
\begin{eqnarray}
l_{N}[u;\lambda ] &=&\left( 
\begin{array}{ccccc}
\lambda u_{N-1,x} & u_{N,x} & 0 & ... & 0 \\ 
0 & \lambda u_{N-1,x} & 2u_{N,x} & \ddots & ... \\ 
... & \ddots & \ddots & \ddots & 0 \\ 
0 & ... & 0 & \lambda u_{N-1,x} & (N-1)u_{N,x} \\ 
-N\lambda ^{N} & -\lambda ^{N-1}Nu_{1,x} & ... & -\lambda ^{2}Nu_{N-2,x} & 
\lambda (1-N)u_{N-1,x}%
\end{array}%
\right) ,\text{ }  \notag \\
&&  \notag \\
q_{N}(\lambda ) &=&\left( 
\begin{array}{ccccc}
0 & 0 & 0 & 0 & 0 \\ 
-\lambda & 0 & 0 & 0 & 0 \\ 
0 & -\lambda & \ddots & ... & 0 \\ 
0 & 0 & \ddots & 0 & 0 \\ 
0 & 0 & 0 & -\lambda & 0%
\end{array}%
\right) ,\ \   \label{r06a}
\end{eqnarray}%
in a very simple and useful for applications form for any $N\in \mathbb{Z}%
_{+},$ which is scale equivalent to that found by Popowicz \ \cite{Po}. \
Moreover, \ we prove that this Riemann type hydrodynamical hierarchy
generates the bi-Hamiltonian flows on the manifold $M_{N}$ and, finally, we
analyze for the cases $N=\overline{2,3}$ and $N=4$ the corresponding
compatible \cite{Bl,PM,FT} Poissonian structures, following within the
gradient-holonomic scheme from these new Lax type representations. A \
mathematical nature of the Lax type representations \ (\ref{r03}),\ (\ref%
{r05}) \ presents from the differential-algebraic point of view \ a very
interesting question, an answer to which may be useful for the integrability
theory, remains open and needs additional investigations.

It is also worth to mention that the methods devised in this work can be
successfully applied to other interesting for applications \cite%
{Wh,Os,Va,BS,DHH} nonlinear dynamical systems, such as Burgers, Korteweg-de
Vries and Ostrovsky-Vakhnenko equations 
\begin{equation}
D_{t}u=D_{x}^{2}u,\text{ \ \ }D_{t}u=D_{x}^{3}u,\text{ \ }D_{x}D_{t}u=-u,
\label{R06b}
\end{equation}%
on the $2\pi $-periodic functional manifold $M_{1}\subset C^{(\infty )}(%
\mathbb{R}/2\pi \mathbb{Z};\mathbb{R)},$ new infinite hierarchies \cite%
{PrYa,BPAP} of Riemann type hydrodynamic systems

\begin{equation}
D_{t}^{N-1}u=D_{x}^{4}\bar{z},\text{ \ \ \ }D_{t}\bar{z}=0,  \label{R06c}
\end{equation}%
and%
\begin{equation}
D_{t}^{N-1}u=\bar{z}_{x}^{2},\text{ \ \ \ \ }D_{t}\bar{z}=0  \label{R06d}
\end{equation}%
on a smooth $2\pi $-periodic functional manifold $M_{N}\subset $ $C^{(\infty
)}(\mathbb{R}/2\mathbb{\pi Z};\mathbb{R}^{N})$ for $N\in \mathbb{N}\ $\ and
many others.

\section{Differential-algebraic approach revisiting}

We will consider the ring $\mathcal{K}:=\mathbb{R}\{\{x,t\}\},$ $%
(x,t)^{\intercal }\in \mathbb{R}^{2},$ of the convergent germs of
real-valued smooth functions from $C^{\infty }(\mathbb{R}^{2};\mathbb{R})$
and construct \cite{Kap, Rit, Ko,Wei,CH} the associated differential
polynomial ring $\mathcal{K}\{u\}:=\mathcal{K}[\Theta u]$ with respect to a
functional variable $u_{1}:=u\in \mathcal{K},$ \ where $\Theta $ denotes the
standard monoid of \ the all commuting differentiations $\partial /\partial
x $ and $\partial /\partial t.$ The ideal $I\{u\}\subset \mathcal{K}\{u\}$
is called differential if the condition $I\{u\}=\Theta I\{u\}$ holds.

In the differential ring $\mathcal{K}\{u\},$ interpreted $\ $as an invariant
differential ideal in $\mathcal{K},$ there are naturally defined two
differentiations 
\begin{equation}
D_{t},\ D_{x}:\ \mathcal{K}\{u\}\rightarrow \mathcal{K}\{u\},  \label{r07}
\end{equation}%
satisfying the Lie-commutator relationship (\ref{r04}). \ For a general
function $u\in \mathcal{K}$ \ \ there exists the only representation of (\ref%
{r04}) in the ideal $\mathcal{K}\{u\}$ of form (\ref{r03a}). Nonetheless, if
the additional constraint \ (\ref{r01}) holds, its linear finite dimensional
matrix representation in the corresponding functional vector space $\mathcal{%
K}\{u\}^{N}$ for $N\in \mathbb{Z}_{+},$ related with some finitely generated
invariant differential ideal $\mathcal{I}\{u\}\subset \mathcal{K}\{u\},$ may
exists being thereby equivalent to the corresponding Lax type representation
for the Riemann type dynamical system (\ref{r02}). To make this scheme
analytically feasible, we consider in detail the cases $N=\overline{1,4}$
and construct the corresponding linear finite dimensional matrix
representations of the Lie-commutator relationship (\ref{r04}), polynomially
depending on an arbitrary spectral parameter $\lambda \in \mathbb{C}.$

\subsection{The case $N=1$}

Aiming to find the corresponding representation vector space for the
Lie-algebraic relationship (\ref{r04}) at $N=1,$ we need firstly to
construct \cite{PAPP} a so called invariant generating Riemann differential
ideal $R\{u\}\subset \mathcal{K}\{u\}$ as 
\begin{equation}
R\{u\}:=\{\sum\limits_{j\in \mathbb{Z}_{+}}\sum\limits_{n\in \mathbb{Z}%
_{+}}\lambda ^{-(n+j)}f_{n}^{(j+1)}D_{t}^{j}D_{x}^{n}u\in \mathcal{R}\{u\}:\
f_{n}^{(j+1)}\in \mathcal{K},\ j,n\in \mathbb{Z}_{+}\},  \label{r09}
\end{equation}%
where $\lambda \in \mathbb{R}\backslash \{0\}$ is an arbitrary parameter. \
The differential ideal \ (\ref{r09}) is, evidently, invariant and
characterized by the following lemma.

\begin{lemma}
\label{Lm_01} The kernel $\mathrm{Ker}D_{t}\subset R\{u\}$ of the
differentiation $D_{t}:\mathcal{K}\{u\}\rightarrow \mathcal{K}\{u\}$ is
generated by elements $f^{(j)}\in \mathbb{R}\{\{x,t\}\},j\in \mathbb{Z}_{+},$
satisfying the linear differential-functional relationships 
\begin{equation}
D_{t}f^{(j+1)}=-\lambda f^{(j)},  \label{r10}
\end{equation}%
where, by definition, 
\begin{equation}
f^{(j+1)}:=f^{(j+1)}(\lambda )=\sum\limits_{n\in \mathbb{Z}%
_{+}}f_{n}^{(j+1)}\lambda ^{-n}  \label{r11}
\end{equation}%
for $\lambda \in \mathbb{R}\backslash \{0\}$ and $j\in \mathbb{Z}_{+}.$
\end{lemma}

Taking now into account the invariant reduction of the differential ideal (%
\ref{r09}) subject to the condition $D_{t}u=0$ to the ideal 
\begin{equation}
R^{(1)}\{u\}:=\{\sum\limits_{n\in \mathbb{Z}_{+}}\lambda
^{-n}f_{n}^{(1)}D_{x}^{n}u\in \mathcal{R}\{u\}:\ D_{t}u=0,\ f_{n}^{(1)}\in 
\mathcal{K},\ n\in \mathbb{Z}_{+}\}  \label{r12}
\end{equation}%
one can simultaneously reduce the hierarchy of relationships (\ref{r10}) to
the simple expression

\begin{equation}
D_{t}f^{(1)}(\lambda )=0,  \label{r13}
\end{equation}%
where $f^{(1)}\in \mathcal{K}_{1}\{u\}:=\mathcal{K}\{u\}\}|_{D_{t}u=0}.$
Now, to formulate the Lax type integrability condition for the case $N=1,$
it is necessary to construct the corresponding $D_{t}$-invariant Lax
differential ideal 
\begin{equation}
L^{(1)}\{u\}:=\left\{ g_{1}f^{(1)}(\lambda )\in \mathcal{K}_{1}\{u\}:\
D_{t}f^{(1)}(\lambda )=0,\ g_{1}\in \mathcal{K}\right\}  \label{r14}
\end{equation}%
and to check its invariance subject to the differentiation $D_{x}:\mathcal{K}%
_{1}\{u\}\rightarrow \mathcal{K}_{1}\{u\}.$ Namely, the following lemma
holds.

\begin{lemma}
\label{Lm_2} The Lax differential ideal (\ref{r14}) is $D_{x}$-invariant, if
the linear equality 
\begin{equation}
D_{x}f^{(1)}(\lambda )=(\lambda u_{x}+\mu u_{x}^{-2}u_{xx})f^{(1)}(\lambda )
\label{r15}
\end{equation}%
holds for arbitrary parameters $\lambda ,\mu \in \mathbb{R},$ where the
subscript $"x"$ means the usual $D_{x}$-differentiation.
\end{lemma}

\begin{proof}
The condition (\ref{r15}) makes the Lax ideal (\ref{r14}) to be also $D_{x}$%
-invariant if the following condition on the linear representation of the $%
D_{x}$-differentiation 
\begin{equation}
D_{x}f^{(1)}(\lambda )=l_{1}[u;\lambda ]f^{(1)}(\lambda )  \label{r16}
\end{equation}%
holds: 
\begin{equation}
D_{t}l_{1}[u;\lambda ]+l_{1}[u;\lambda ]D_{x}u_{1}=0  \label{r17}
\end{equation}%
for any parameter $\lambda \in \mathbb{R}.$ Differential-functional equation
(\ref{r17}) upon substitution 
\begin{equation}
l_{1}[u;\lambda ]:=D_{x}a_{1}[u;\lambda ]  \label{r18}
\end{equation}%
reduces easily to the equation 
\begin{equation}
D_{t}a_{1}[u;\lambda ]=\lambda  \label{r19}
\end{equation}%
for any parameter $\lambda \in \mathbb{R}.$ Taking into account the
condition $D_{t}u_{1}=0$ one can easily find that 
\begin{equation}
a_{1}[u;\lambda ]=\lambda x+\mu /D_{x}u_{1},  \label{r20}
\end{equation}%
where $\mu \in \mathbb{R}$ is arbitrary. Having now substituted expression (%
\ref{r20}) into (\ref{r18}) and taking into account that $D_{t}x=u_{1},$ the
result (\ref{r15}) follows.
\end{proof}

Thereby, having naturally extended the differential ring $\mathcal{K}$ to
the ring $\mathbb{C}\{\{x,t\}\},$ one can formulate the following
proposition.

\begin{proposition}
\label{Prop_1} The Lax type representation for the Riemann hydrodynamical
equation 
\begin{equation}
D_{t}u=0  \label{r21}
\end{equation}%
is given by a set of the linear compatible equations 
\begin{equation}
D_{t}f^{(1)}(\lambda )=0,\ \ D_{x}f^{(1)}(\lambda )=(\lambda u_{x}+\mu
u_{x}^{-2}u_{xx})f^{(1)}(\lambda )  \label{r22}
\end{equation}%
for $f^{(1)}(\lambda )\in L_{\infty }(\mathbb{R}^{2};\mathbb{C})$ and for
any $\lambda ,\mu \in \mathbb{C}.$
\end{proposition}

\subsection{The case $N=2$}

For the case $N=2$ the respectively reduced invariant Riemann differential
ideal \ (\ref{r09}) is given by the set 
\begin{eqnarray}
R^{(2)}\{u\} &:&=\{\sum\limits_{j=\overline{0,1}}\sum\limits_{n\in \mathbb{Z}%
_{+}}\lambda ^{-(n+j)}f_{n}^{(j+1)}D_{t}^{j}D_{x}^{n}u:D_{t}^{2}u=0,\  
\notag \\
\ f_{n}^{(j+1)} &\in &\mathcal{K},\ n\in \mathbb{Z}_{+},\ j=\overline{0,1}\}
\label{r23}
\end{eqnarray}%
and characterized by the kernel $\mathrm{Ker}D_{t}\subset R^{(2)}\{u\}$ of
the $D_{t}$-differentiation, which is generated by the following linear
differential relationships: 
\begin{equation}
D_{t}f^{(1)}(\lambda )=0,\ \ D_{t}f^{(2)}(\lambda )=-\lambda f^{(1)}(\lambda
),  \label{r24}
\end{equation}%
where $f^{(j)}(\lambda )\in \mathcal{K}_{2}\{u\}:=\mathcal{K}%
\{u\}|_{D_{t}^{2}u=0},$ $j=\overline{0,1},$ and $\lambda \in \mathbb{R}$ is
an arbitrary parameter.

The condition (\ref{r24}) can be rewritten in a compact vector form as 
\begin{equation}
D_{t}f(\lambda )=q_{2}(\lambda )f(\lambda ),\ \ q_{2}(\lambda ):=\left( 
\begin{array}{cc}
0 & 0 \\ 
-\lambda & 0%
\end{array}%
\right) ,  \label{r25}
\end{equation}%
where $f(\lambda ):=(f^{(1)}(\lambda ),f^{(2)}(\lambda ))^{\intercal }\in 
\mathcal{K}_{2}\{u\}^{2}.$

Now we can construct the invariant Lax differential ideal 
\begin{eqnarray}
L^{(2)}\{u\} &:&=\{\sum\limits_{j=1}^{2}g_{j}f^{(j)}(\lambda )\in \mathcal{K}%
_{2}\{u\}:\ D_{t}f^{(1)}(\lambda )=0,\ \   \label{r26} \\
D_{t}f^{(2)}(\lambda ) &=&-\lambda f^{(1)}(\lambda ),g_{j}\in \mathcal{K},\
j=\overline{1,2}\}  \notag
\end{eqnarray}%
and to check its invariance subject to a linear representation of the $D_{x}$%
-differentiation in the form: 
\begin{equation}
D_{x}f(\lambda )=l_{2}[u;\lambda ]f(\lambda )  \label{r27}
\end{equation}%
for some matrix $l_{2}[u;\lambda ]\in \mathrm{End}$ $\mathbb{R}^{2}.$

The following lemma holds.

\begin{lemma}
The Lax differential ideal (\ref{r26}) is $D_{x}$-invariant if the matrix 
\begin{equation}
l_{2}[u;\lambda ]=\left( 
\begin{array}{cc}
\lambda u_{1,x} & u_{2,x} \\ 
-2\lambda ^{2} & -\lambda u_{1,x}%
\end{array}%
\right) .  \label{r28}
\end{equation}
\end{lemma}

\begin{proof}
The $D_{x}$-invariant condition for the Lax ideal (\ref{r26}) forces the
matrix $l_{2}[u;\lambda ]\in \mathrm{End}$ $\mathbb{R}^{2}$ to satisfy the
differential-functional relationship 
\begin{equation}
D_{t}l_{2}+l_{2}D_{x}u_{1}=[q_{2},l_{2}],  \label{r29}
\end{equation}%
where we have put by definition $l_{2}:=l_{2}[u;\lambda ],$ $%
q_{2}:=q_{2}(\lambda )\in \mathrm{End}$ $\mathbb{R}^{2}$. Making the
substitution 
\begin{equation}
l_{2}=D_{x}a_{2}  \label{r30}
\end{equation}%
for some matrix $a_{2}:=a_{2}[u;\lambda ]\in \mathrm{End}$ $\mathbb{R}^{2},$
we can easily reduce equation (\ref{r29}) to the equivalent one in the form 
\begin{equation}
D_{t}a_{2}=[q_{2},a_{2}].  \label{r31}
\end{equation}%
To solve equation (\ref{r31}) it is useful to take into account that 
\begin{equation}
D_{t}^{2}a_{2}=k_{2}q_{2},\ \ \ \ D_{t}^{3}a_{2}=0  \label{r32}
\end{equation}%
for some constant $k_{2}\in \mathbb{R}.$ Taking into account the
relationship (\ref{r32}) one easily obtains the exact matrix representation 
\begin{equation}
a_{2}=\left( 
\begin{array}{cc}
\lambda u_{1} & u_{2} \\ 
-2\lambda ^{2}x & -\lambda u_{1}%
\end{array}%
\right) ,  \label{r33}
\end{equation}%
entailing the result (\ref{r28})
\end{proof}

Thus, having as above extended the differential ring $\mathcal{K}$ to the
ring $\mathbb{C}\{\{x,t\}\},$ one can formulate the next proposition.

\begin{proposition}
\label{Prop_2} The Lax representation for the Riemann type hydrodynamical
system 
\begin{equation}
D_{t}u_{1}=u_{2},\ \ \ \ D_{t}u_{2}=0  \label{r34}
\end{equation}%
is given by a set of the linear compatible equations 
\begin{equation}
D_{t}f(\lambda )=q_{2}(\lambda ):=\left( 
\begin{array}{cc}
0 & 0 \\ 
-\lambda & 0%
\end{array}%
\right) f(\lambda ),\ \ \ D_{x}f(\lambda )=\left( 
\begin{array}{cc}
\lambda u_{1,x} & u_{2,x} \\ 
-2\lambda ^{2} & -\lambda u_{1,x}%
\end{array}%
\right) f(\lambda )  \label{r34_1}
\end{equation}%
for $f(\lambda )\in L_{\infty }(\mathbb{R}^{2};\mathbb{C}^{2})$ and
arbitrary complex parameter $\lambda \in \mathbb{C}.$
\end{proposition}

\subsection{The case $N=3$}

Similarly to the above, for the case $N=3$ the respectively reduced
invariant Riemann differential ideal \ (\ref{r09}) is given by the set 
\begin{eqnarray}
R^{(3)}\{u\} &:&=\{\sum\limits_{j=\overline{0,2}}\sum\limits_{n\in \mathbb{Z}%
_{+}}\lambda ^{-(j+n)}f_{n}^{(j+1)}D_{t}^{j}D_{x}^{n}u:\ D_{t}^{3}u=0, 
\notag \\
\ f_{n}^{(j+1)} &\in &\mathcal{K},\ n\in \mathbb{Z}_{+},\ j=\overline{0,2}\}
\label{r35}
\end{eqnarray}%
and characterized by the kernel $\mathrm{Ker}$ $D_{t}\subset R^{(3)}\{u\},$
generated by the following differential relationships: 
\begin{eqnarray}
&&D_{t}f^{(1)}(\lambda )=0,\ \ \ D_{t}f^{(2)}(\lambda )=-\lambda
f^{(1)}(\lambda ),  \label{r36} \\
&&D_{t}f^{(3)}(\lambda )=-\lambda f^{(2)}(\lambda ),  \notag
\end{eqnarray}%
where $f^{(j)}(\lambda )\in \mathcal{K}_{3}\{u\}:=\mathcal{K}%
\{u\}|_{D_{t}^{3}u=0},$ $j=\overline{0,2},$ and $\lambda \in \mathbb{R}$ is
an arbitrary parameter. The system (\ref{r36}) can be equivalently rewritten
in the matrix form 
\begin{equation}
D_{t}f(\lambda )=q_{3}(\lambda )f(\lambda ),\ \ q_{3}(\lambda )=\left( 
\begin{array}{ccc}
0 & 0 & 0 \\ 
-\lambda & 0 & 0 \\ 
0 & -\lambda & 0%
\end{array}%
\right) ,  \label{r37}
\end{equation}%
where $f(\lambda ):=(f^{(1)},f^{(2)},f^{(3)}(\lambda ))^{\intercal }\in 
\mathcal{K}_{3}\{u\}^{3}.$

Now we can construct the corresponding invariant Lax differential ideal 
\begin{eqnarray}
L_{3}\{u\} &:&=\{\sum\limits_{j=1}^{3}g_{j}f^{(j)}(\lambda )\in \mathcal{K}%
_{3}\{u\}:D_{t}f^{(1)}(\lambda )=0,\ \ \ D_{t}f^{(2)}(\lambda )=-\lambda
f^{(1)}(\lambda ),  \label{r38} \\
\ D_{t}f^{(3)}(\lambda ) &=&-\lambda f^{(2)}(\lambda ),g_{j}\in \mathcal{K}%
,\ j=\overline{1,3}\},  \notag
\end{eqnarray}%
whose $D_{x}$-invariance is looked for in the linear matrix form 
\begin{equation}
D_{x}f(\lambda )=l_{3}[u;\lambda ]f(\lambda )  \label{r39}
\end{equation}%
for some matrix $l_{3}[u;\lambda ]\in \mathrm{End}$ $\mathbb{R}^{3}.$ The
latter satisfies the differential-functional equation 
\begin{equation}
D_{t}l_{3}+l_{3}D_{x}u_{1}=[q_{3},l_{3}],  \label{r40}
\end{equation}%
where we put, by definition, $l_{3}:=l_{3}[u;\lambda ],$ $%
q_{3}:=q_{3}(\lambda )\in \mathrm{End}$ $\mathbb{R}^{3}.$ Making use of the
derivative representation 
\begin{equation}
l_{3}=D_{x}a_{3},  \label{r41}
\end{equation}%
where $a_{3}:=a_{3}[u;\lambda ]\in \mathrm{End}$ $\mathbb{R}^{2},$ equation (%
\ref{r40}) reduces to an equivalent one in the form 
\begin{equation}
D_{t}a_{3}=[q_{3},a_{3}].  \label{r42}
\end{equation}%
To solve matrix equation (\ref{r42}), we take into account that 
\begin{equation}
D_{t}^{3}a_{3}=k_{3}q_{3},\ \ D_{t}^{4}a_{3}=0  \label{r43}
\end{equation}%
for some constant $k_{3}\in \mathbb{R}.$ Based both on (\ref{r43}) and on
the component wise form of (\ref{r42}) one easily finds that 
\begin{equation}
a_{3}[u;\lambda ]=\left( 
\begin{array}{ccc}
\lambda u_{2} & u_{3} & 0 \\ 
0 & \lambda u_{2} & 2u_{3} \\ 
-3\lambda ^{3}x & -3\lambda ^{2}u_{1} & -2\lambda u_{2}%
\end{array}%
\right) ,  \label{r44}
\end{equation}%
entailing the matrix 
\begin{equation}
l_{3}[u;\lambda ]=\left( 
\begin{array}{ccc}
\lambda u_{2,x} & u_{3,x} & 0 \\ 
0 & \lambda u_{2,x} & 2u_{3,x} \\ 
-3\lambda ^{3} & -3\lambda ^{2}u_{1,x} & -2\lambda u_{2,x}%
\end{array}%
\right) .  \label{r45}
\end{equation}%
Thereby, having naturally extended the differential ring $\mathcal{K}$ to
the ring $\mathbb{C}\{\{x,t\}\},$ one can formulate the next proposition.

\begin{proposition}
\label{Prop_3} The Lax representation for the Riemann type hydrodynamical
system 
\begin{equation}
D_{t}u_{1}=u_{2},\ \ D_{t}u_{2}=u_{3},\ \ D_{t}u_{3}=0  \label{r46}
\end{equation}%
is given by a set of linear compatible equations 
\begin{equation}
D_{t}f(\lambda )=\left( 
\begin{array}{ccc}
0 & 0 & 0 \\ 
-\lambda & 0 & 0 \\ 
0 & -\lambda & 0%
\end{array}%
\right) f(\lambda ),\ \ D_{x}f(\lambda )=\left( 
\begin{array}{ccc}
\lambda u_{2,x} & u_{3,x} & 0 \\ 
0 & \lambda u_{2,x} & 2u_{3,x} \\ 
-3\lambda ^{3} & -3\lambda ^{2}u_{1,x} & -2\lambda u_{2,x}%
\end{array}%
\right) f(\lambda ),  \label{r47}
\end{equation}%
where $f(\lambda )\in L_{\infty }(\mathbb{R}^{2};\mathbb{C}^{3})$ and $%
\lambda \in \mathbb{C}$ is an arbitrary complex parameter.
\end{proposition}

As one can easily observe, the scheme of finding the Lax representation for
the Riemann type hydrodynamical equations at $N=\overline{1,3}$ can be
easily generalized for arbitrary $N\in \mathbb{Z}_{+},$ that is a topic of
the next section.

\subsection{The general case $N\in \mathbb{Z}_{+}$}

Consider at arbitrary $N\in \mathbb{Z}_{+}$ the constraint $D_{t}^{N}u=0$
and the respectively reduced invariant Riemann differential ideal \ (\ref%
{r09}), which is given by the set 
\begin{eqnarray}
R^{(N)}\{u\} &:&=\{\sum\limits_{j=\overline{0}}^{N-1}\sum\limits_{n\in 
\mathbb{Z}_{+}}\lambda ^{-(j+n)}f_{n}^{(j+1)}D_{t}^{j}D_{x}^{n}u\in \mathcal{%
K}\{u\}:\ D_{t}^{N}u=0.  \notag \\
f_{n}^{(j+1)} &\in &\mathbb{R}\{\{x,t\}\},\ n\in \mathbb{Z}_{+},\ j=%
\overline{0,N-1}\}.  \label{r48}
\end{eqnarray}%
The related kernel $\mathrm{Ker}D_{t}\subset R^{(N)}\{u\}$ of the $D_{t}$%
-differentiation is, owing to (\ref{r10}) and \ (\ref{r48}), generated by
the relationships 
\begin{equation}
D_{t}f^{(1)}(\lambda )=0,\ D_{t}f^{(2)}(\lambda )=-\lambda f^{(1)}(\lambda ),%
\text{ }...,\ D_{t}f^{(N)}(\lambda )=-\lambda f^{(N-1)}(\lambda ),
\label{r49}
\end{equation}%
which can be rewritten in a matrix form as 
\begin{equation}
D_{t}f(\lambda )=q_{N}(\lambda )f(\lambda ),\ \ q_{N}(\lambda )=\left( 
\begin{array}{ccccc}
0 & 0 & 0 & 0 & 0 \\ 
-\lambda & 0 & 0 & 0 & 0 \\ 
0 & -\lambda & \ddots & ... & 0 \\ 
0 & 0 & \ddots & 0 & 0 \\ 
0 & 0 & 0 & -\lambda & 0%
\end{array}%
\right) ,  \label{r50}
\end{equation}%
where, by definition, a vector $f(\lambda
):=(f^{(1)},f^{(2)},...,f^{(N)})^{\intercal }\in \mathcal{K}_{N}\{u\}^{N}:=%
\mathcal{K}\{u\}^{N}|_{D_{t}^{N}u=0}.$ The latter generates the invariant
Lax differential ideal 
\begin{equation}
L^{(N)}\{u\}:=\{\sum\limits_{j=1}^{N}g_{j}f^{(j)}(\lambda )\in \mathcal{K}%
_{N}\{u\}:\ g_{j}\in \mathcal{K},\ j=\overline{1,N}\},  \label{r51}
\end{equation}%
whose $D_{x}$-invariance holds, if its linear matrix representation in the
space $\ \mathcal{K}^{(N)}\{u\}^{N}$ 
\begin{equation}
D_{x}f(\lambda )=l_{N}[u;\lambda ]f(\lambda )  \label{r52}
\end{equation}%
is compatible with (\ref{r49}) for some matrix $l_{N}:=l_{N}[u;\lambda ]\in 
\mathrm{End}$ $\mathbb{R}^{N}$ and arbitrary $\lambda \in \mathbb{R}.$ As a
result we obtain the compatibility condition 
\begin{equation}
D_{t}l_{N}+l_{N}D_{x}u=[q_{N},l_{N}].  \label{r53}
\end{equation}%
Having represented the matrix $l_{N}\in \mathrm{End}$ $\mathbb{R}^{N}$ in
the derivative form 
\begin{equation}
l_{N}:=D_{x}a_{N},  \label{r54}
\end{equation}%
we can reduce the functional-differential equation (\ref{r53}) to 
\begin{equation}
D_{t}a_{N}=[q_{N},a_{N}].  \label{r55}
\end{equation}%
Now, taking into account that 
\begin{equation}
D_{t}^{N}a_{N}=k_{N}q_{N},\ \ D_{t}^{N+1}a_{N}=0  \label{r56}
\end{equation}%
for some constant $k_{N}\in \mathbb{R},$ one can easily obtain by means of
simple calculations the following solution to equation (\ref{r55}): 
\begin{equation}
a_{N}[u;\lambda ]=\left( 
\begin{array}{ccccc}
\lambda u_{N-1} & u_{N} & 0 & ... & 0 \\ 
0 & \lambda u_{N-1} & 2u_{N} & \ddots & ... \\ 
... & \ddots & \ddots & \ddots & 0 \\ 
0 & ... & 0 & \lambda u_{N-1} & (N-1)u_{N-1} \\ 
-N\lambda ^{N}x & -\lambda ^{N-1}Nu_{1} & ... & -\lambda ^{2}Nu_{N-2} & 
\lambda (1-N)u_{N-1}%
\end{array}%
\right) ,  \label{r57}
\end{equation}%
entailing the exact matrix expression 
\begin{equation}
l_{N}[u;\lambda ]=\left( 
\begin{array}{ccccc}
\lambda u_{N-1,x} & u_{N,x} & 0 & ... & 0 \\ 
0 & \lambda u_{N-1,x} & 2u_{N,x} & \ddots & ... \\ 
... & \ddots & \ddots & \ddots & 0 \\ 
0 & ... & 0 & \lambda u_{N-1,x} & (N-1)u_{N,x} \\ 
-N\lambda ^{N} & -\lambda ^{N-1}Nu_{1,x} & ... & -\lambda ^{2}Nu_{N-2,x} & 
\lambda (1-N)u_{N-1,x}%
\end{array}%
\right) .  \label{r58}
\end{equation}%
Thereby, having naturally extended the differential ring $\mathcal{K}$ to
the ring $\mathbb{C}\{\{x,t\}\}\},$ one can formulate the following general
proposition.

\begin{proposition}
\label{Prop_4} The Lax representation for the generalized Riemann type
hydrodynamical system 
\begin{equation}
D_{t}u_{1}=u_{2},\ D_{t}u_{2}=u_{3},...,\ D_{t}u_{N-1}=u_{N},\ D_{t}u_{N}=0
\label{r59}
\end{equation}%
is given for any arbitrary $N\in \mathbb{Z}_{+}$ by a set of linear
compatible equations 
\begin{eqnarray}
D_{t}f(\lambda ) &=&\left( 
\begin{array}{ccccc}
0 & 0 & 0 & 0 & 0 \\ 
-\lambda & 0 & 0 & 0 & 0 \\ 
0 & -\lambda & \ddots & ... & 0 \\ 
0 & 0 & \ddots & 0 & 0 \\ 
0 & 0 & 0 & -\lambda & 0%
\end{array}%
\right) f(\lambda ),\ \   \label{r60} \\
D_{x}f(\lambda ) &=&\left( 
\begin{array}{ccccc}
\lambda u_{N-1,x} & u_{N,x} & 0 & ... & 0 \\ 
0 & \lambda u_{N-1,x} & 2u_{N,x} & \ddots & ... \\ 
... & \ddots & \ddots & \ddots & 0 \\ 
0 & ... & 0 & \lambda u_{N-1,x} & (N-1)u_{N,x} \\ 
-N\lambda ^{N} & -\lambda ^{N-1}Nu_{1,x} & ... & -\lambda ^{2}Nu_{N-2,x} & 
\lambda (1-N)u_{N-1,x}%
\end{array}%
\right) f(\lambda ),  \notag
\end{eqnarray}%
where $f(\lambda )\in L_{\infty }(\mathbb{R}^{2};\mathbb{C}^{N})$ and $%
\lambda \in \mathbb{C}$ is an arbitrary complex parameter. \ Moreover, the
relationships (\ref{r60}) realize a linear matrix representation of the
commutator condition (\ref{r04}) in the vector space $\mathcal{K}%
^{(N)}\{u\}^{N}.$
\end{proposition}

The general Lax type representation (\ref{r60}), obtained above for
arbitrary $N\in \mathbb{Z}_{+},$ looks essentially simpler than those
obtained before for $N=\overline{1,4}$ in \cite{PoP,PAPP,GBPPP} and given by
differential-matrix expressions (\ref{r05}), depending at $N\geqslant 3$ on
the solutions to the set of differential-functional equations (\ref{r06}).
As it was already mentioned above, it was first found in a similar form up
to a scaling parameter $\lambda \in \mathbb{C}$ in \cite{Po} \ by means of \
another approach. Taking this simplicity into account, one can apply to the
Lax type representation (\ref{r60}) the gradient-holonomic approach and find
the corresponding bi-Hamiltonian structures, responsible for the Lax type
integrability of the Riemann type hierarchy of hydrodynamical systems (\ref%
{r59}).

\section{The bi-Hamiltonian structures for cases $N=\overline{2,3}$ and $N=4$%
}

To proceed with proving the related bi-Hamiltonian integrability of the
generalized Riemann type dynamical system \ (\ref{r59}) at arbitrary $N\in 
\mathbb{Z}_{+}$ and with finding the naturally associated with the Lax
representations (\ref{r60}) Poissonian structures, we need preliminarily to
study with respect to the gradient-holonomic approach \cite{PP,PM,BPS} the
corresponding spectral properties of the linear differential system (\ref%
{r52}). Below will analyze the next cases: $N=\overline{2,3}$ and $N=4.$

\subsection{ The case $N=2$}

Making use of the exact matrix expression (\ref{r28}), one can define for
the second linear equation of (\ref{r34_1}) the corresponding monodromy
matrix $S(x;\lambda )\in \mathrm{End}$ $\mathbb{C}^{3},x\in \mathbb{R},$
which satisfies \cite{No,MBPS,PM,BPS,HPP} the Novikov-Marchenko commutator
equation 
\begin{equation}
D_{x}S(x;\lambda )=[l_{2}[u;\lambda ],S(x;\lambda )],  \label{r61}
\end{equation}%
Since the trace-functional $\gamma (\lambda ):=\mathrm{tr}S(x;\lambda )$ is,
owing to the Lax representation (\ref{r47}), invariant with respect to the $%
\mathbb{R}\ni x,t$-evolutions on the manifold $M_{2},$ its gradient $\varphi
(x;\lambda ):=\mathrm{grad}\gamma (\lambda )\in T^{\ast }(M_{2}),$ $x\in 
\mathbb{R},$ depending parametrically on the spectral parameter $\lambda \in 
\mathbb{C},$ equals 
\begin{equation}
\varphi (x;\lambda )=\mathrm{tr}(l_{2}^{\prime \ast }S(x;\lambda ))
\label{r62}
\end{equation}%
and satisfies the well known gradient relationship 
\begin{equation}
\vartheta \varphi (x;\lambda )=z(\lambda )\eta \varphi (x;\lambda )
\label{r63}
\end{equation}%
for some meromorphic function $z:\mathbb{C\rightarrow C}$ and $x\in \mathbb{R%
},$ where $\vartheta ,\eta :T^{\ast }(M_{2})\rightarrow T(M_{2})$ are
compatible \cite{Bl,PM,FT} Poissonian structures on the manifold $M_{2},$ \
the dash sign \textquotedblleft $^{\prime }$\textquotedblleft\ means the
usual Frechet derivative and \textquotedblleft $\ast $\textquotedblleft\
means the corresponding conjugation with respect to the standard bilinear
form on $T^{\ast }(M_{2})\times T(M_{2}).$ The latter \ naturally follows
from the Lie-algebraic theory \cite{FT,Bl,PM} of Lax type integrable
dynamical systems, treating them as the corresponding Lie-Poisson flows on
the adjoint space $\mathcal{G}^{\ast }(\lambda )$ to a suitably constructed
centrally extended affine Lie algebra $\mathcal{G}(\lambda ),\lambda \in 
\mathbb{C}.$ By simple enough calculations one find that%
\begin{equation}
\varphi (x;\lambda ):=(\varphi _{1},\varphi _{2})^{\intercal }=(\lambda
(S_{x}^{(11)}-S_{x}^{(22)}),S_{x}^{(21)})^{\intercal },  \label{r64}
\end{equation}%
where we put, by definition, 
\begin{equation}
S(x;\lambda ):=\left( 
\begin{array}{cc}
S^{(11)} & S^{(12)} \\ 
S^{(21)} & S^{(22)}%
\end{array}%
\right) .  \label{r65}
\end{equation}%
Taking into account that equation \ (\ref{r61}) can be rewritten as the
system 
\begin{eqnarray}
S_{x}^{(11)} &=&-u_{2,x}S^{(21)}+2\lambda ^{2}S^{(12)},  \label{r66} \\
S_{x}^{(12)} &=&-u_{2,x}(S^{(11)}-S^{(22)})+2\lambda u_{1,x}S^{(12)},  \notag
\\
S_{x}^{(21)} &=&-2\lambda ^{2}u_{1,x}(S^{(11)}-S^{(22)})-2\lambda
u_{1,x}S^{(21)},  \notag \\
S_{x}^{(22)} &=&-2\lambda ^{2}S^{(12)}+u_{2,x}S^{(21)},  \notag
\end{eqnarray}%
by means of simple enough calculations, consisting in substituting (\ref{r64}%
) into (\ref{r66}), one easily finds the gradient relationship \ (\ref{r63})
in the exact differential-matrix form:%
\begin{equation}
\left( 
\begin{array}{cc}
0 & -1 \\ 
1 & 0%
\end{array}%
\right) \left( 
\begin{array}{c}
\varphi _{1} \\ 
\varphi _{2}%
\end{array}%
\right) =2\lambda \left( 
\begin{array}{cc}
\partial ^{-1} & u_{1,x}\partial ^{-1} \\ 
\partial ^{-1}u_{1,x} & u_{2,x}\partial ^{-1}+\partial ^{-1}u_{2,x}%
\end{array}%
\right) \left( 
\begin{array}{c}
\varphi _{1} \\ 
\varphi _{2}%
\end{array}%
\right) ,  \label{r67}
\end{equation}%
where, by definition, $z(\lambda ):=2\lambda ,$ $\lambda \in \mathbb{C},$ and

\begin{equation}
\vartheta :=\left( 
\begin{array}{cc}
0 & 1 \\ 
-1 & 0%
\end{array}%
\right) ,\eta :=\left( 
\begin{array}{cc}
\partial ^{-1} & u_{1,x}\partial ^{-1} \\ 
\partial ^{-1}u_{1,x} & u_{2,x}\partial ^{-1}+\partial ^{-1}u_{2,x}%
\end{array}%
\right)  \label{r68}
\end{equation}%
constitute the corresponding compatible Poissonian structures on the
functional manifold $M_{2}.$ Thus, one can formulate the following
proposition, before stated in \cite{Pav,BD,GBPPP,PoP} using different
approaches.

\begin{proposition}
\label{Prop_5}\bigskip The Riemann type hydrodynamical system \ (\ref{r02})
\ at $N=2$ is a Lax integrable bi-Hamiltonian flow with respect to the
compatible on the functional manifold $M_{2}$ Poissonian pair \ (\ref{r68}).
\end{proposition}

As a consequence of Proposition \ (\ref{Prop_5}) easily one states that
there exists \cite{Bl,FT,PM} an associated infinite hierarchy of commuting
to each other integrable bi-Hamiltonian flows on $M_{2}$ 
\begin{equation}
d(u_{1},u_{2})^{\intercal }/dt_{n}:=-(\eta \vartheta
^{-1})^{n}(u_{1,x},u_{2,x})^{\intercal },  \label{r69}
\end{equation}%
where $t_{n}\in \mathbb{R},n\in \mathbb{Z}_{+},$ are the corresponding
evolution parameters.

\subsection{ The cases $N=3$ and $4$}

For the case $N=3$ there was proved in \cite{GPPP,PoP,GBPPP} that the
corresponding Riemann type hydrodynamical system (\ref{r46}) is a
Hamiltonian dynamical system on the $2\pi $-periodic functional manifold $%
M_{3}:=\{(u_{1},u_{2},u_{3})^{\intercal }\in C^{(\infty )}(\mathbb{R}/(2\pi 
\mathbb{Z});\mathbb{R}^{3})\},$ whose exact Poissonian structure is given by
the differential-matrix expression 
\begin{equation}
\eta =\left( 
\begin{array}{ccc}
\partial ^{-1} & u_{1,x}\partial ^{-1} & 0 \\ 
\partial ^{-1}u_{1,x} & u_{2,x}\partial ^{-1}+\partial ^{-1}u_{2,x} & 
\partial ^{-1}u_{3,x} \\ 
0 & u_{3,x}\partial ^{-1} & 0%
\end{array}%
\right) ,  \label{r61a}
\end{equation}%
acting as a linear mapping $\eta :\ T^{\ast }(M_{3})\rightarrow T(M_{3})$
from the cotangent space $T^{\ast }(M_{3})$ to the tangent space\ $\
T(M_{3}).$ Making use of the exact matrix expression (\ref{r45}), one can
define for the linear equation (\ref{r39}) the corresponding monodromy
matrix $S(x;\lambda )\in \mathrm{End}$ $\mathbb{C}^{3},x\in \mathbb{R},$
which satisfies the Novikov-Marchenko commutator equation 
\begin{equation}
dS(x;\lambda )/dx=[l_{3}[u;\lambda ],S(x;\lambda )].  \label{r62a}
\end{equation}%
Since the trace-functional $\gamma (\lambda ):=\mathrm{tr}S(x;\lambda )$ is
invariant with respect to the $\mathbb{R}\ni x,t$-evolutions on the manifold 
$M_{3},$ its gradient $\varphi (x;\lambda ):=\mathrm{grad}\gamma (\lambda
)\in T^{\ast }(M_{3}),$ $x\in \mathbb{R},$ depending parametrically on the
spectral parameter $\lambda \in \mathbb{C},$ equals 
\begin{equation}
\varphi (x;\lambda )=\mathrm{tr}(l_{3}^{\prime \ast }S(x;\lambda
))=(-3\lambda ^{2}S_{x}^{(23)},\lambda
(S_{x}^{(11)}+S_{x}^{(22)}-2S_{x}^{(23)},S_{x}^{(21)}+2S_{x}^{(32)})^{%
\intercal }  \label{r63a}
\end{equation}%
and satisfies the gradient relationship 
\begin{equation}
\vartheta \varphi (x;\lambda )=z(\lambda )\eta \varphi (x;\lambda )
\label{r64a}
\end{equation}%
for some meromorphic function $z:\mathbb{C\rightarrow C}$ and all $x\in 
\mathbb{R},$ where $\vartheta :T^{\ast }(M_{3})\rightarrow T(M_{3})$ is a
second compatible with \ (\ref{r61a}) Poisson structure on the manifold $%
M_{3}.$ \ In this case the standard calculations, consisting in substituting
(\ref{r63}) into (\ref{r62}) and reducing the result to the form \ (\ref{r63}%
), give rise by means of slightly cumbersome calculations to the equivalent
to \ (\ref{r64a}) relationship%
\begin{equation}
\varphi (x;\lambda )=\lambda ^{2}\vartheta ^{-1}\eta \varphi (x;\lambda ),
\label{r64b}
\end{equation}%
where 
\begin{equation}
\vartheta ^{-1}=\left( 
\begin{array}{ccc}
\partial u_{3}+u_{3}\partial & -u_{2,x}-\partial u_{2} & -2u_{1}\partial
+2\partial \frac{u_{3}u_{1,x}}{u_{3,x}} \\ 
u_{2,x}-u_{2}\partial & \partial u_{1}+u_{1}\partial & -2-2\partial \frac{%
u_{3}}{u_{3,x}} \\ 
-2\partial u_{1}+2\frac{u_{3}u_{1,x}}{u_{3,x}}\partial & 2-2\frac{u_{3}}{%
u_{3,x}}\partial & 
\begin{array}{c}
-(\frac{u_{3}u_{1,x}^{2}+2u_{3}u_{2,x}}{u_{3,x}^{2}})\partial - \\ 
-\partial (\frac{u_{3}u_{1,x}^{2}+2u_{3}u_{2,x}}{u_{3,x}^{2}})%
\end{array}%
\end{array}%
\right)  \label{r64c}
\end{equation}%
is the second compatible with \ (\ref{r61a}) co-Poissonian structure on the
functional manifold $M_{3}.$

A completely similar results hold in the case $N=4.$ Namely, the
corresponding Novikov-Marchenko equation%
\begin{equation}
dS(x;\lambda )/dx=[l_{4}[u;\lambda ],S(x;\lambda )].  \label{r70}
\end{equation}%
jointly with the expression%
\begin{eqnarray}
\varphi (x;\lambda ) &=&\mathrm{tr}(l_{4}^{\prime \ast }S(x;\lambda
))=(-4\lambda ^{3}S_{x}^{(24)},-4\lambda ^{2}S_{x}^{(34)},  \label{r71} \\
&&\lambda
(S_{x}^{(11)}+S_{x}^{(22)}+S_{x}^{(33)}-3S_{x}^{(44)}),S_{x}^{(21)}+2S_{x}^{(32)}+3S^{(43)})^{\intercal },
\notag
\end{eqnarray}%
is reduced by means of simple but cumbersome \ calculations to the gradient
relationship\ 
\begin{equation}
\varphi (x;\lambda )=\lambda ^{3}\vartheta ^{-1}\eta \varphi (x;\lambda ),
\label{r71a}
\end{equation}%
where the Poissonian structures $\ \ \eta $ and $\vartheta :T^{\ast
}(M_{4})\rightarrow T(M_{4})$ are, respectively, inverse to the
co-Poissonian operators 
\begin{equation}
\eta ^{-1}=\left( 
\begin{array}{cccc}
0 & 0 & -\partial & \partial \frac{u_{3,x}}{u_{4,x}} \\ 
0 & \partial & 0 & -\partial \frac{u_{2,x}}{u_{4,x}} \\ 
-\partial & 0 & 0 & \partial \frac{u_{1,x}}{u_{4,x}} \\ 
\frac{u_{3,x}}{u_{4,x}}\partial & -\frac{u_{2,x}}{u_{4,x}}\partial & \frac{%
u_{1,x}}{u_{4,x}}\partial & 
\begin{array}{c}
\frac{1}{2}[u_{4,x}^{-2}(u_{2,x}^{2}-2u_{1,x}u_{3,x})\partial + \\ 
+\partial (u_{2,x}^{2}-2u_{1,x}u_{3,x})u_{4,x}^{-2}]%
\end{array}%
\end{array}%
\right)  \label{r72}
\end{equation}%
and 
\begin{equation}
\vartheta ^{-1}=\left( 
\begin{array}{cccc}
0 & -3u_{4,x} & u_{3}\partial & 
\begin{array}{c}
\begin{array}{c}
-\partial u_{2}-u_{2}\partial - \\ 
-u_{3}\partial \frac{u_{3,x}}{u_{4,x}}- \\ 
-\frac{u_{3,x}}{u_{4,x}}\ \partial u_{3}%
\end{array}
\\ 
\end{array}
\\ 
3u_{4,x} & 0 & u_{2}\partial & 
\begin{array}{c}
\\ 
\begin{array}{c}
\partial u_{1}+u_{1}\partial + \\ 
+u_{2}\partial \frac{u_{3,x}}{u_{4,x}}+ \\ 
+\frac{u_{3,x}}{u_{4,x}}\ \partial u_{2}%
\end{array}
\\ 
\end{array}
\\ 
\partial u_{3} & \partial u_{2} & -u_{1}\partial -u_{1}\partial & 
\begin{array}{c}
\\ 
\begin{array}{c}
-5+\partial (\frac{u_{4}-u_{2}u_{2,x}}{u_{4,x}})- \\ 
-\partial \frac{(u_{1}u_{3})_{x}}{u_{4,x}}-u_{1}\partial \frac{u_{3,x}}{%
u_{4,x}}%
\end{array}
\\ 
\end{array}
\\ 
\begin{array}{c}
-\partial u_{2}-u_{2}\partial - \\ 
-u_{3}\partial \frac{u_{3,x}}{u_{4,x}}- \\ 
-\frac{u_{3,x}}{u_{4,x}}\ \partial u_{3}%
\end{array}
& 
\begin{array}{c}
\partial u_{1}+u_{1}\partial + \\ 
+u_{2}\partial \frac{u_{3,x}}{u_{4,x}}+ \\ 
+\frac{u_{3,x}}{u_{4,x}}\ \partial u_{2}%
\end{array}
& 
\begin{array}{c}
5+(\frac{u_{4}-u_{2}u_{2,x}}{u_{4,x}})\partial - \\ 
-\frac{(u_{1}u_{3})_{x}}{u_{4,x}}\partial -\frac{u_{3,x}}{u_{4,x}}\partial
u_{1}%
\end{array}
& 
\begin{array}{c}
\\ 
\begin{array}{c}
\frac{u_{4}-u_{2}u_{2,x}}{u_{4,x}^{2}}\partial + \\ 
-\frac{u_{3,x}}{u_{4,x}^{2}}\partial -\partial \frac{u_{3,x}}{u_{4,x}^{2}}-
\\ 
-\partial \frac{u_{4}-u_{2}u_{2,x}}{u_{4,x}^{2}}%
\end{array}%
\end{array}%
\end{array}%
\right) ,  \label{r73}
\end{equation}%
on the functional manifold $M_{4}.$ Below we will reanalyze the
bi-Hamiltonian structures of the generalized Riemann type hierarchy (\ref%
{r01}) from the geometric point of view, described in detail in \cite%
{PM,BPS,PoP}.

\subsection{Poissonian structures: the geometric approach}

To construct the Hamiltonian structures, related with the general dynamical
system (\ref{r02}), by means of the geometric approach \ it is first
necessary to present it in the following equivalent form:%
\begin{equation}
\left. 
\begin{array}{c}
\begin{array}{c}
du_{1}/dt=u_{2}-u_{1}u_{1,x} \\ 
du_{2}/dt=u_{3}-u_{1}u_{2,x} \\ 
... \\ 
du_{j}/dt=u_{j+1}-u_{1}u_{j,x} \\ 
... \\ 
du_{N}/dt=-u_{1}u_{N,x}%
\end{array}%
\end{array}%
\right\} :=K[u],\text{ }  \label{H33a}
\end{equation}
where $K:M_{N}\rightarrow T(M_{N})$ is the corresponding vector field on the
functional manifold $M_{N}$ for a fixed $N\in \mathbb{Z}_{+}.$ Then the
following proposition \cite{PoP,PM,BPS} holds.

\begin{proposition}
\label{Prop_6}Let $\psi \in T^{\ast }(M_{N})$ be, in general, a quasi-local
functional-analytic vector, which satisfies the condition $\psi ^{\prime
}\neq \psi ^{\prime ,\ast }$ and solves the Lie-Lax equation%
\begin{equation}
L_{K}\psi :=\psi _{t}+K^{\prime ,\ast }\psi =\mathrm{grad}\text{ }\mathcal{L}%
,\text{ \ \ }  \label{H33}
\end{equation}%
for some smooth functional $\mathcal{L}\in D(M_{N}).$ Then the dynamical
system (\ref{H33a}) allows the Hamiltonian representation%
\begin{equation}
\begin{array}{c}
K[u]=-\eta \text{ }\mathrm{grad}\text{ }H_{(\eta )}[u],\text{ } \\ 
H_{\eta }=(\psi _{(\eta )},K)-\mathcal{L}_{(\eta )},\text{ \ \ }\eta
^{-1}=\psi _{(\eta )}^{\prime }-\psi _{(\eta )}^{\prime ,\ast },%
\end{array}
\label{H34}
\end{equation}%
where $\eta :T^{\ast }(M_{N})\rightarrow T(M_{N})$ is the corresponding
Poissonian structure, invertible \ on $M_{N}.$ Otherwise, if the operator $%
\eta ^{-1}:T(M_{N})\rightarrow T^{\ast }(M_{N})$ is not invertible, the
following relationship 
\begin{equation}
\mathrm{grad}\text{ }H_{(\eta )}[u]=-\eta ^{-1}K[u]  \label{H34aa}
\end{equation}%
holds on the manifold $M_{N}.$
\end{proposition}

\begin{remark}
Concerning a general solution to \ the Lie-Lax equation (\ref{H33}) it is
easy to observe that it possesses the following representation:%
\begin{equation}
\psi =\bar{\psi}+\mathrm{grad}\text{ }\mathcal{\bar{L}},  \label{H34aa1}
\end{equation}%
where 
\begin{equation}
L_{K}\bar{\psi}:=\bar{\psi}_{t}+K^{\prime ,\ast }\bar{\psi}=0,\text{ \ }L_{K}%
\mathcal{\bar{L}=L}.\text{\ \ }  \label{H34aa2}
\end{equation}%
The related co-Poissonian operator 
\begin{equation}
\eta ^{-1}:=\psi ^{\prime }-\psi ^{\prime ,\ast }=\bar{\psi}^{\prime }-\bar{%
\psi}^{\prime ,\ast },  \label{H34aa3}
\end{equation}
owing to the Volterra symmetry condition $(\mathrm{grad}$ $\mathcal{\bar{L}}%
)^{\prime }=(\mathrm{grad}$ $\mathcal{\bar{L}})^{\prime ,\ast },$ persists
to be unchangeable.
\end{remark}

Based on Proposition \ref{Prop_6}, it is enough to search for non-symmetric
functional-analytic solutions to \ (\ref{H33}) making use for this, for
instance, of special computer-algebraic algorithms.

\subsubsection{The case $N=3$}

The corresponding dynamical system 
\begin{equation}
\left. 
\begin{array}{c}
\begin{array}{c}
du_{1}/dt=u_{2}-u_{1}u_{1,x} \\ 
du_{2}/dt=u_{3}-u_{1}u_{2,x} \\ 
du_{3}/dt=-u_{1}u_{3,x}%
\end{array}%
\end{array}%
\right\} :=K[u]  \label{H34ab}
\end{equation}%
is defined on the functional manifold $M_{3}.$ A vector $\psi :=(\psi
_{1},\psi _{2},\psi _{3})^{\intercal }\in T^{\ast }(M_{3})$ satisfies owing
to \ (\ref{H33}) a set of functional-differential equations%
\begin{eqnarray}
d\psi _{1}/dt+u_{1}\psi _{1,x}-u_{2,x}\psi _{2}-u_{3,x}\psi _{3} &=&\delta 
\mathcal{L}/\delta u_{1},  \notag \\
d\psi _{2}/dt+\psi _{1}-(u_{1}\psi _{2)x} &=&\delta \mathcal{L}/\delta u_{2},
\label{H34a} \\
d\psi _{3}/dt+\psi _{2}+(u_{1}\psi _{3})_{x} &=&\delta \mathcal{L}/\delta
u_{3}  \notag
\end{eqnarray}%
for some smooth functional $\mathcal{L}\in \mathcal{D}(M_{3}).$ System \ (%
\ref{H34a}) can be slightly simplified, if to use the differentiation $%
D_{t}:M_{3}\rightarrow T(M_{3}):$ 
\begin{eqnarray}
D_{t}\psi _{1}-u_{2,x}\psi _{2}-u_{3,x}\psi _{3} &=&\delta \mathcal{L}%
/\delta u_{1}  \label{H34b} \\
D_{t}\psi _{2}+\psi _{1}+u_{1,x}\psi _{2} &=&\delta \mathcal{L}/\delta u_{2},
\notag \\
D_{t}\psi _{3}+\psi _{2}+u_{1,x}\psi _{3} &=&\delta \mathcal{L}/\delta u_{3},
\notag
\end{eqnarray}%
which is equivalent to the scalar functional-analytic expression%
\begin{equation}
D_{t}^{2}(\alpha \psi _{1})=D_{t}(\alpha \delta \mathcal{L}/\delta
u_{1})+(D_{t}\alpha )\delta \mathcal{L}/\delta u_{1}+(D_{t}^{2}\alpha
)\delta \mathcal{L}/\delta u_{2}+\delta \mathcal{L}/\delta u_{3}
\label{H34c}
\end{equation}%
on the alone function $\psi _{1}\in \mathcal{K}\{u\},$ where we have denoted
the function $\alpha :=u_{3,x}^{-1},$ satisfying the useful differential
relationships 
\begin{equation}
D_{t}\alpha =u_{1,x}\alpha ,\text{ \ \ \ }D_{t}^{2}\alpha =u_{2,x}\alpha ,%
\text{ \ \ \ }D_{t}^{3}\alpha =1,\text{ \ \ }D_{t}^{4}\alpha =0.
\label{H34cc}
\end{equation}%
The corresponding functions $\psi _{2},\psi _{3}$ $\in \mathcal{K}\{u\}$ are
given by the functional-operator expressions%
\begin{eqnarray}
\psi _{2} &=&\alpha ^{-1}D_{t}^{-1}(-\alpha \psi _{1}+\alpha \delta \mathcal{%
L}/\delta u_{2}),  \label{H34d} \\
\psi _{3} &=&\alpha ^{-1}D_{t}^{-1}(-\alpha \psi _{2}+\alpha \delta \mathcal{%
L}/\delta u_{3},  \notag
\end{eqnarray}%
easily following from \ (\ref{H34b}). In the special case $\mathcal{L}=0$
equation \ (\ref{H34c}) reduces to 
\begin{equation}
D_{t}^{2}(\alpha \psi _{1})=0,  \label{H34e}
\end{equation}%
whose functional-analytic solutions can be found analytically \ by means of
both differential-algebraic tools, devised in \cite{PAPP}, and modern
computer-algebraic algorithms. In particular, making use of the
differential-algebraic approach of the work \cite{PAPP}, one can easily
enough to find, amongst many others, the following solutions to (\ref{H34c}%
): 
\begin{eqnarray}
\psi _{1} &=&u_{1,x}/2,\text{ \ \ \ }\mathcal{L}_{(\eta )}=\frac{1}{2}%
\int_{0}^{2\pi }(2u_{3}+u_{2}u_{1,x})dx;  \label{H34f} \\
&&  \notag \\
\psi _{1} &=&u_{1,x}u_{3}-u_{1}u_{3,x},\text{ \ \ }\mathcal{L}_{(\vartheta
)}=0;  \notag \\
&&  \notag
\end{eqnarray}%
and%
\begin{eqnarray}
\psi _{1} &=&-u_{3,x}/2,\text{ \ \ }\mathcal{L}_{(\vartheta _{0})}=0;
\label{H34ff} \\
&&  \notag \\
\psi _{1} &=&u_{2}u_{3,x},\text{ \ \ \ }\mathcal{L}_{(\vartheta _{1})}=0; 
\notag \\
&&  \notag \\
\psi _{1} &=&u_{1,x}u_{2}-2u_{3}-(\frac{u_{1}u_{2}}{u_{3}}-\frac{u_{2}^{3}}{%
3u_{3}^{2}})u_{3,x},\text{ \ \ \ }\mathcal{L}_{(\vartheta _{2})}=0;  \notag
\\
&&  \notag \\
\psi _{1} &=&u_{1,x}-\frac{u_{2}^{2}u_{3,x}}{2u_{3}},\text{ \ \ \ }\mathcal{L%
}_{(\vartheta _{3})}=0;  \notag \\
&&  \notag \\
\psi _{1} &=&u_{1,x}-\frac{u_{1}u_{3,x}}{u_{3}},\text{ \ \ \ }\mathcal{L}%
_{(\vartheta _{4})}=0;  \notag
\end{eqnarray}%
giving rise to the generating vectors%
\begin{equation}
\begin{array}{c}
\psi _{(\eta
)}=(u_{1,x}/2,0,-u_{3,x}^{-1}u_{1,x}^{2}/2+u_{3,x}^{-1}u_{2,x})^{\intercal },%
\text{ \ \ \ }\mathcal{L}_{(\eta )}=\frac{1}{2}\int_{0}^{2\pi
}(2u_{3}+u_{2}u_{1,x})dx; \\ 
\\ 
\psi _{(\vartheta )}=(u_{1,x}u_{3}-u_{1}u_{3,x},\text{ \ }%
u_{1}u_{2,x}-u_{1,x}u_{2},\text{ \ \ }2u_{2}-u_{1}u_{1,x}+\frac{%
u_{3}u_{1,x}^{2}-2u_{3}u_{2,x}}{u_{3,x}})^{\intercal },\text{ \ \ \ }%
\mathcal{L}_{(\vartheta )}=0; \\ 
\\ 
\psi _{(\theta )}=(u_{2}-\frac{u_{2}u_{2,x}^{3}}{6u_{3,x}}+\frac{%
u_{3}u_{2,x}^{4}}{24u_{3,x}^{3}},\frac{u_{2}u_{2,x}}{u_{3,x}}-\frac{%
u_{2}u_{1,x}^{2}}{2u_{3,x}}-\frac{u_{3}u_{2,x}^{2}}{2u_{3,x}^{2}}-\frac{%
u_{3}u_{1,x}u_{2,x}^{3}}{6u_{3,x}^{3}}+\frac{u_{3}u_{2,x}^{5}}{30u_{3,x}^{4}}%
, \\ 
\\ 
-\frac{u_{2}u_{2,x}^{2}}{u_{3,x}^{2}}+\frac{u_{2}u_{1}^{2}u_{2,x}}{%
u_{3,x}^{2}}+\frac{u_{3}u_{2,x}^{2}}{u_{3,x}^{2}}-\frac{%
u_{2}u_{1,x}u_{2,x}^{3}}{u_{3,x}^{3}}+\frac{u_{3}u_{1,x}u_{2,x}^{4}}{%
u_{3,x}^{4}}-\frac{u_{3}u_{2,x}^{6}}{30u_{3,x}^{5}})^{\intercal },\text{ \ \ 
}\mathcal{L}_{(\theta )}=0;%
\end{array}
\label{H34g}
\end{equation}%
and 
\begin{equation}
\begin{array}{c}
\psi _{(\vartheta _{0})}=(-u_{3,x}/2,\text{ \ }u_{2,x}/2,\text{ \ \ }%
u_{1,x}/2-u_{3,x}^{-1}u_{2,x}^{2}/2)^{\intercal },\text{ \ \ \ }\mathcal{L}%
_{(\vartheta _{0})}=0; \\ 
\\ 
\psi _{(\vartheta _{1})}=(u_{2}u_{3,x},\text{ \ }-u_{1}u_{3,x},\text{ \ \ }%
u_{1}u_{2}u_{3,x}u_{3}^{-1}-u_{2}^{3}u_{3,x}u_{3}^{-2}/3)^{\intercal },\text{
\ }\mathcal{L}_{(\vartheta _{1})}=0; \\ 
\end{array}
\label{H35}
\end{equation}%
\begin{equation*}
\begin{array}{c}
\psi _{(\vartheta _{2})}=(u_{1,x}u_{2}-2u_{3}-(\frac{u_{1}u_{2}}{u_{3}}-%
\frac{u_{2}^{3}}{3u_{3}^{2}})u_{3,x},-u_{2}-\frac{u_{1}^{2}u_{2,x}}{2u_{3}}+%
\frac{u_{3,x}u_{2}^{4}}{12u_{3}^{3}}-\frac{u_{3}u_{2,x}}{u_{3,x}}+\frac{%
u_{3}u_{1,x}^{2}}{2u_{3,x}}, \\ 
\\ 
u_{1}-\frac{u_{1}u_{1,x}u_{2,x}}{u_{3,x}}+\frac{u_{1}^{2}u_{2}u_{3,x}}{%
2u_{3}^{2}}-\frac{u_{1}u_{2}^{3}u_{3,x}}{3u_{3}^{3}}+\frac{u_{3}u_{1,x}}{%
u_{3,x}}+\frac{u_{2}^{5}u_{3,x}}{15u_{3}^{4}}-\frac{2u_{3}u_{1,x}^{2}u_{2,x}%
}{u_{3,x}^{2}}+ \\ 
\\ 
+\frac{u_{2}u_{1,x}^{2}}{2u_{3,x}}+\frac{u_{1}u_{2,x}^{3}}{3u_{3,x}^{2}}-%
\frac{u_{3}u_{1,x}u_{2,x}}{2u_{3,x}^{2}}-\frac{u_{2}u_{2,x}^{4}}{%
12u_{3,x}^{3}}+\frac{u_{3}u_{1,x}u_{2,x}^{3}}{3u_{3,x}^{3}}-\frac{%
u_{3}u_{2,x}^{5}}{15u_{3,x}^{4}})^{\intercal },\mathcal{L}_{(\vartheta
_{2})}=0; \\ 
\\ 
\psi _{(\vartheta _{3})}=(u_{1,x}-\frac{u_{2}^{2}u_{3,x}}{2u_{3}^{2}},\frac{%
u_{2}^{3}u_{3,x}-6u_{3}^{3}}{6u_{3}^{3}},\frac{%
(u_{2}^{4}-6u_{1}^{2}u_{3}^{2})u_{3,x}}{12u_{3}^{4}}+\frac{u_{2}-u_{1}u_{1,x}%
}{u_{3}}+\frac{u_{2}u_{2,x}(3u_{1}u_{3}-u_{2}^{2})}{3u_{3}})^{\intercal },%
\mathcal{L}_{(\vartheta _{3})}=0; \\ 
\\ 
\psi _{(\vartheta _{4})}=(u_{1,x}-\frac{u_{1}u_{3,x}}{u_{3}},\text{ \ \ \ }%
\frac{u_{1}u_{2,x}-u_{1,x}u_{2}}{u_{3}},\frac{u_{1}u_{1,x}-2u_{2}}{u_{3}}+%
\frac{2u_{2,x}-u_{1,x}^{2}}{u_{3,x}})^{\intercal },\text{ \ \ \ \ \ \ \ \ \ }%
\mathcal{L}_{(\vartheta _{4})}=0;%
\end{array}%
\end{equation*}%
As a result of \ (\ref{H34}) and \ expressions (\ref{H34g}) one easily
obtains the corresponding co-Poissonian operators: 
\begin{equation}
\eta ^{-1}:=\psi _{(\eta )}^{\prime }-\psi _{(\eta )}^{\prime ,\ast }=\left( 
\begin{array}{ccc}
\partial & 0 & -\partial u_{1,x}u_{3,x}^{-1} \\ 
0 & 0 & \partial u_{3,x}^{-1} \\ 
-u_{1x}u_{3,x}^{-1}\partial & u_{3,x}^{-1}\partial & 
\begin{array}{c}
\frac{1}{2}(u_{1,x}^{2}u_{3,x}^{-2}\partial +\partial
u_{1,x}^{2}u_{3,x}^{-2})- \\ 
-(u_{2,x}u_{3,x}^{-2}\partial +\partial u_{2,x}u_{3,x}^{-2})%
\end{array}%
\end{array}%
\right) ,  \label{H35a}
\end{equation}%
\begin{equation*}
\vartheta ^{-1}:=\psi _{(\vartheta )}^{\prime }-\psi _{(\vartheta )}^{\prime
,\ast }=\left( 
\begin{array}{ccc}
u_{3}\partial +\partial u_{3} & -u_{2,x}-\partial u_{2} & -2u_{1}\partial
+2\partial \frac{u_{3}u_{1,x}}{u_{3,x}} \\ 
u_{2,x}-u_{2}\partial & u_{1}\partial +\partial u_{1} & -2-2\partial \frac{%
u_{3}}{u_{3,x}} \\ 
-2\partial u_{1}+2\frac{u_{3}u_{1,x}}{u_{3,x}}\partial & 2-2\frac{u_{3}}{%
u_{3,x}}\partial & 
\begin{array}{c}
-(\frac{u_{3}u_{1,x}^{2}+2u_{3}u_{2,x}}{u_{3,x}^{2}})\partial - \\ 
-\partial (\frac{u_{3}u_{1,x}^{2}+2u_{3}u_{2,x}}{u_{3,x}^{2}})%
\end{array}%
\end{array}%
\right) ,
\end{equation*}%
and 
\begin{equation}
\vartheta _{0}^{-1}:=\psi _{(\vartheta _{0})}^{\prime }-\psi _{(\vartheta
_{0})}^{\prime ,\ast }=\left( 
\begin{array}{ccc}
0 & 0 & 0 \\ 
0 & \partial & -\partial u_{2,x}u_{3,x}^{-1} \\ 
0 & -u_{2,x}u_{3,x}^{-1}\partial & \frac{1}{2}(u_{2,x}^{2}u_{3,x}^{-2}%
\partial +\partial u_{2,x}^{2}u_{3,x}^{-2})%
\end{array}%
\right) ,  \label{H35c}
\end{equation}%
\begin{equation*}
\vartheta _{1}^{-1}:=\psi _{(\vartheta _{1})}^{\prime }-\psi _{(\vartheta
_{1})}^{\prime ,\ast }=\left( 
\begin{array}{ccc}
0 & 2u_{3,x} & u_{2}\partial -u_{2}u_{3,x}/u_{3} \\ 
-2u_{3,x} & 0 & 
\begin{array}{c}
-u_{1}\partial -u_{1}u_{3,x}/u_{3}+ \\ 
+u_{2}^{2}u_{3,x}/u_{3}^{2}%
\end{array}
\\ 
u_{2}u_{3,x}/u_{3}+\partial u_{2} & 
\begin{array}{c}
u_{1}u_{3,x}/u_{3}- \\ 
-u_{2}^{2}u_{3,x}/u_{3}^{2}-\partial u_{1}%
\end{array}
& 
\begin{array}{c}
u_{1}u_{2}\partial u_{3}^{-1}+u_{3}^{-1}\partial u_{1}u_{2}- \\ 
-\frac{u_{2}^{3}}{3}\partial \frac{1}{u_{3}^{2}}-\frac{1}{3u_{3}^{2}}%
\partial u_{2}^{3}%
\end{array}%
\end{array}%
\right) ,
\end{equation*}%
and so on. The second expression \ of (\ref{H35a}) coincides exactly with
the co-Poissonian operator \ (\ref{r64c}), satisfying the gradient
relationship \ (\ref{r64b}), which proves the bi-Hamiltonicity of the
Riemann type equation \ (\ref{H34ab}). Concerning the next two expressions \
(\ref{H35a}) and \ (\ref{H35c}) it is easy to observe that \ they are not
strictly invertible, since the translation vector field $d/dx:M_{3}%
\rightarrow T(M_{3})$ with components $(u_{1,x},u_{2,x},u_{3,x})^{\intercal
}\in T(M_{3})$ belongs to their kernels, that is $\eta
^{-1}(u_{1,x},u_{2,x},u_{3,x})^{\intercal }=$ $0$ $=\vartheta
_{0}^{-1}(u_{1,x},u_{2,x},u_{3,x})^{\intercal }.$ Nonetheless, upon formal
inverting the operator expression (\ref{H35a}), we obtain by means of simple
enough, but slightly cumbersome, direct calculations, that the Poissonian $%
\eta $-operator looks as

\begin{equation}
\eta =\left( 
\begin{array}{ccc}
\partial ^{-1} & u_{1,x}\partial ^{-1} & 0 \\ 
\partial ^{-1}u_{1,x} & u_{2,x}\partial ^{-1}+\partial ^{-1}u_{2,x} & 
\partial ^{-1}u_{3,x} \\ 
0 & u_{3,x}\partial ^{-1} & 0%
\end{array}%
\right) ,  \label{H36}
\end{equation}%
coinciding with expression \ (\ref{r61a}), and the corresponding Hamiltonian
function equals 
\begin{equation}
H_{(\eta )}:=(\psi _{(\eta )},K)-\mathcal{L}_{(\eta )}=\int_{0}^{2\pi
}dx(u_{1,x}u_{2}-u_{3}),  \label{H37}
\end{equation}%
Concerning the operator expression\ (\ref{H35c}) the corresponding
Hamiltonian function%
\begin{equation}
H_{(\vartheta _{0})}:=(\psi _{(\vartheta _{0})},K)=\int_{0}^{2\pi
}dx(u_{3}u_{2,x})  \label{H38}
\end{equation}%
It is evident that these Hamiltonian functions \ (\ref{H37}) and \ (\ref{H38}%
) are conservation laws for the dynamical system \ (\ref{H34ab}), which were
also before found in \cite{PoP}.

\begin{remark}
It \ is worth to mention here that as the sum of vectors $%
\sum_{j=0}^{4}s_{j}\psi _{(\vartheta _{j})}\in T^{\ast }(M_{3})$ with
arbitrary $s_{j}\in \mathbb{R},j=\overline{0,4},$ solves the determining
equation \ (\ref{H33}), \ the corresponding operator 
\begin{equation}
\theta ^{-1}:=\sum_{j=0}^{4}s_{j}\vartheta _{j}^{-1}  \label{H39}
\end{equation}%
will also be a co-Poissonian operator for the dynamical system \ (\ref{H34ab}%
).
\end{remark}

\subsubsection{\protect\bigskip The case $N=4$}

The Riemann type hydrodynamical equation (\ref{r02}) at $N=4$ is equivalent
to the nonlinear dynamical system%
\begin{equation}
\left. 
\begin{array}{c}
u_{1,t}=u_{2}-u_{1}u_{1,x} \\ 
u_{2,t}=u_{3}-u_{1}u_{2,x} \\ 
u_{3,t}=u_{4}-u_{1}u_{3,x} \\ 
u_{4,t}=-u_{1}u_{4,x}%
\end{array}%
\right\} :=K[u_{1},u_{2},u_{3},u_{4}],  \label{H41}
\end{equation}%
where $K:M_{4}\rightarrow T(M_{4})$ is a suitable vector field on the smooth 
$2\pi $-periodic functional manifold $M_{4}:=$ $C^{\infty }(\mathbb{R}/2\pi 
\mathbb{Z};\mathbb{R}^{4}\mathbb{)}.$ To study its possible Hamiltonian
structures, we need to find functional solutions to the determining Lie-Lax
equation (\ref{H33}):%
\begin{equation}
\psi _{t}+K^{^{\prime ,\ast }}\psi =grad\text{ }\mathcal{L}  \label{H42}
\end{equation}%
for $\psi \in T^{\ast }(M_{4})$ and some smooth functional $\mathcal{L}\in
D(M_{4}),$ where 
\begin{equation}
K^{^{\prime }}=\left( 
\begin{array}{cccc}
-\partial u_{1} & 1 & 0 & 0 \\ 
-u_{2,x} & -u_{1}\partial & 1 & 0 \\ 
-u_{3,x} & 0 & -u_{1}\partial & 1 \\ 
-u_{4_{,}x} & 0 & 0 & -u_{1}\partial%
\end{array}%
\right) ,\text{ \ }K^{^{\prime ,\ast }}=\left( 
\begin{array}{cccc}
u_{1}\partial & -u_{2,x} & -u_{3,x} & -u_{4,x} \\ 
1 & \partial u_{1} & 0 & 0 \\ 
0 & 1 & \partial u_{1} & 0 \\ 
0 & 0 & 1 & \partial u_{1}%
\end{array}%
\right)  \label{H43}
\end{equation}%
are, respectively, the Frechet derivative of the mapping $K:M_{4}\rightarrow
T(M_{4})$ and its conjugate. \ The small parameter method \cite{PM,PoP},
applied to equation (\ref{H42}), gives rise to the following its exact
solution:%
\begin{equation}
\psi _{(\eta )}=(-u_{3,x},u_{2,x}/2,0,-\frac{u_{2,x}^{2}}{2u_{4,x}}+\frac{%
u_{1,x}u_{3,x}}{u_{4,x}})^{\intercal },\mathcal{L}_{(\eta )}=\int_{0}^{2\pi
}(u_{4}u_{1,x}-u_{2}u_{3,x}/2)dx.  \label{H44}
\end{equation}%
As a result, we obtain right away from (\ref{H34}) that dynamical system (%
\ref{H41}) is a Hamiltonian system on the functional manifold $M_{4},$ that
is 
\begin{equation}
K=-\vartheta \text{ }grad\text{ }H_{(\eta )},  \label{H45}
\end{equation}%
where the Hamiltonian functional equals 
\begin{equation}
H_{(\eta )}:=(\psi _{(\eta )},K)-\mathcal{L}_{(\eta )}=\int_{0}^{2\pi
}(u_{1}u_{4,x}-u_{2}u_{3,x})dx  \label{H46}
\end{equation}%
and the co-implectic operator equals 
\begin{equation}
\eta ^{-1}:=\psi _{(\eta )}^{\prime }-\psi _{(\eta )}^{\prime ,\ast }=\left( 
\begin{array}{cccc}
0 & 0 & -\partial & \partial \frac{u_{3,x}}{u_{4,x}} \\ 
0 & \partial & 0 & -\partial \frac{u_{2,x}}{u_{4,x}} \\ 
-\partial & 0 & 0 & \partial \frac{u_{1,x}}{u_{4,x}} \\ 
\frac{u_{3,x}}{u_{4,x}}\partial & -\frac{u_{2,x}}{u_{4,x}}\partial & \frac{%
u_{1,x}}{u_{4,x}}\partial & 
\begin{array}{c}
\frac{1}{2}[u_{4,x}^{-2}(u_{2,x}^{2}-2u_{1,x}u_{3,x})\partial + \\ 
+\partial (u_{2,x}^{2}-2u_{1,x}u_{3,x})u_{4,x}^{-2}]%
\end{array}%
\end{array}%
\right) .  \label{H47}
\end{equation}%
The latter is degenerate: \ the relationship $\eta
^{-1}(u_{1,x},u_{2,x},u_{3,x},u_{4,x})^{\intercal }=0$ exactly on the whole
manifold $M_{4},$ but the inverse to (\ref{H47}) exists and can be
calculated analytically.

Proceed now to constructing other solutions to the generating equation (\ref%
{H42}) at the reduced case $\mathcal{L}=0\mathbf{,}$ having rewritten it in
the following form:%
\begin{eqnarray}
D_{t}\psi _{1}-u_{2,x}\psi _{2}-u_{3,x}\psi _{3}-u_{4,x}\psi _{4} &=&0, 
\notag \\
D_{t}\psi _{2}+\psi _{1}+u_{1,x}\psi _{2} &=&0,  \notag \\
D_{t}\psi _{3}+\psi _{2}+u_{1,x}\psi _{3} &=&0,  \label{H48} \\
D_{t}\psi _{4}+\psi _{3}+u_{1,x}\psi _{4} &=&0,  \notag
\end{eqnarray}%
where $\psi :=(\psi _{1},\psi _{2},\psi _{3},\psi _{4})^{\intercal }\in
T^{\ast }(M_{4}).$ The latter is equivalent to the system of
differential-functional relationships%
\begin{equation}
\begin{array}{c}
D_{t}(\alpha \psi _{1})-\psi _{1}D_{t}\alpha -\psi _{2}D_{t}^{2}\alpha -\psi
_{3}D_{t}^{3}\alpha -\psi _{4}D_{t}^{4}\alpha =0, \\ 
D_{t}(\alpha \psi _{2})+\alpha \psi _{1}=0,D_{t}(\alpha \psi _{3})+\alpha
\psi _{2}=0,D_{t}(\alpha \psi _{4})+\alpha \psi _{4}=0,%
\end{array}
\label{H48a}
\end{equation}%
where we have denoted the function $\alpha :=u_{4,x}^{-1}.$ From (\ref{H48a}%
) one ensues, by means of simple calculations, the following generating
differential-functional equation 
\begin{equation}
D_{t}^{2}(\alpha \psi _{1})=0  \label{H49}
\end{equation}%
on the manifold $M_{4}.$ As above in the case $N=3,$ the obtained equation (%
\ref{H49}) can be easily enough solved by means of the
differential-algebraic approach, devised before in \cite{PoP}, which based
on the following observation: owing to the equality $D_{t}^{5}\alpha =0$ the
set 
\begin{equation}
\mathcal{A}\{u\}:=\{f_{0}\alpha +f_{1}D_{t}\alpha +f_{2}D_{t}^{2}\alpha
+f_{3}D_{t}^{3}\alpha +f_{4}D_{t}^{4}\alpha \in \mathcal{K}%
_{4}\{u\}:f_{j}\in \mathcal{K}_{4}\{u\},j=\overline{0,4}\}  \label{H50}
\end{equation}%
generates a finite dimensional invariant differential \ ideal of the
differential ring $\mathcal{K}_{4}\{u\}:=\mathcal{K}\{u\}|_{D_{t}^{4}u=0}.$
As a result of constructing solutions to the functional-differential
equation (\ref{H49}), as non-symmetric elements of the ideal (\ref{H50}),
one finds \ the following expression:%
\begin{equation}
\begin{array}{c}
\psi _{(\vartheta
)}=(u_{2,x}u_{4}-u_{2}u_{4,x},u_{4,x}u_{1}-u_{4}u_{1,x},2u_{4}-(u_{1}u_{3}-u_{2}^{2}/2)_{x},
\\ 
7u_{3}-u_{2}u_{1,x}+u_{1}u_{2,x}+\frac{u_{3,x}}{u_{4,x}}%
[u_{4}-(u_{1}u_{3}-u_{2}^{2}/2)_{x}])^{\intercal },\text{ \ }\mathcal{L}%
_{(\vartheta )}=0.%
\end{array}
\label{H51}
\end{equation}%
Now taking into account the expression (\ref{H34}) one easily obtains the
second co-Poissonian operator $\vartheta ^{-1}=\psi _{(\vartheta )}^{\prime
}-\psi _{(\vartheta )}^{\prime ,\ast }$ in the following matrix form: 
\begin{eqnarray}
\vartheta ^{-1} &=&\left( 
\begin{array}{cccc}
0 & -3u_{4,x} & u_{3}\partial & 
\begin{array}{c}
\begin{array}{c}
-\partial u_{2}-u_{2}\partial - \\ 
-u_{3}\partial \frac{u_{3,x}}{u_{4,x}}- \\ 
-\frac{u_{3,x}}{u_{4,x}}\ \partial u_{3}%
\end{array}
\\ 
\end{array}
\\ 
3u_{4,x} & 0 & u_{2}\partial & 
\begin{array}{c}
\\ 
\begin{array}{c}
\partial u_{1}+u_{1}\partial + \\ 
+u_{2}\partial \frac{u_{3,x}}{u_{4,x}}+ \\ 
+\frac{u_{3,x}}{u_{4,x}}\ \partial u_{2}%
\end{array}
\\ 
\end{array}
\\ 
\partial u_{3} & \partial u_{2} & -u_{1}\partial -u_{1}\partial & 
\begin{array}{c}
\\ 
\begin{array}{c}
-5+\partial (\frac{u_{4}-u_{2}u_{2,x}}{u_{4,x}})- \\ 
-\partial \frac{(u_{1}u_{3})_{x}}{u_{4,x}}-u_{1}\partial \frac{u_{3,x}}{%
u_{4,x}}%
\end{array}
\\ 
\end{array}
\\ 
\begin{array}{c}
-\partial u_{2}-u_{2}\partial - \\ 
-u_{3}\partial \frac{u_{3,x}}{u_{4,x}}- \\ 
-\frac{u_{3,x}}{u_{4,x}}\ \partial u_{3}%
\end{array}
& 
\begin{array}{c}
\partial u_{1}+u_{1}\partial + \\ 
+u_{2}\partial \frac{u_{3,x}}{u_{4,x}}+ \\ 
+\frac{u_{3,x}}{u_{4,x}}\ \partial u_{2}%
\end{array}
& 
\begin{array}{c}
5+(\frac{u_{4}-u_{2}u_{2,x}}{u_{4,x}})\partial - \\ 
-\frac{(u_{1}u_{3})_{x}}{u_{4,x}}\partial -\frac{u_{3,x}}{u_{4,x}}\partial
u_{1}%
\end{array}
& 
\begin{array}{c}
\\ 
\begin{array}{c}
\frac{u_{4}-u_{2}u_{2,x}}{u_{4,x}^{2}}\partial + \\ 
-\frac{u_{3,x}}{u_{4,x}^{2}}\partial -\partial \frac{u_{3,x}}{u_{4,x}^{2}}-
\\ 
-\partial \frac{u_{4}-u_{2}u_{2,x}}{u_{4,x}^{2}}%
\end{array}%
\end{array}%
\end{array}%
\right) ,  \notag \\
&&  \label{H52}
\end{eqnarray}%
satisfying jointly with (\ref{H47}) the gradient relationship (\ref{r71a}).

\section{Conclusion}

A differential-algebraic approach, elaborated in this article for revisiting
the integrability analysis of generalized Riemann type hydrodynamical
equation (\ref{r01}), made it possible to construct new and more simpler Lax
type representation for the general case $N\in \mathbb{Z}_{+},$ contrary to
those constructed before in \cite{PAPP,PoP}. \ This representation, obtained
by means of the suggested differential-algebraic approach, proves also to
coincide up to the scaling parameter $\lambda \in \mathbb{C}$ with that
first obtained by Z. Popowicz in \cite{Po}. \ It is worth to mention that
the corresponding Lax type representations \ \ \ for the generalized Riemann
type hierarchy are well fitting for constructing the related compatible
Poissonian structures and proving their corresponding bi-Hamiltonian
integrability. The corresponding calculations are fulfilled in details for
case $N=\overline{1,2}$ and preliminary results are obtained for cases $N=%
\overline{3,4}$ by means of the geometric method, devised in \cite%
{PM,BPS,PoP}. It was also demonstrated that the corresponding
differential-functional relationships give rise to the suitable compatible
Poissonian structures and present a very interesting mathematical problem
from the differential-algebraic point of view, which is planned to be
studied \ in other work.\bigskip

\section{Acknowledgements}

The authors are much obliged to Prof. Denis Blackmore (NJIT, New Jersey,
USA) for \ very instrumental discussion of the work, valuable advices,
comments and remarks. Special acknowledgment belongs to the Scientific and
Technological Research Council of Turkey (T\"{U}BITAK-2011) for a partial
support of the research by A.K. Prykarpatsky and Y.A. Prykarpatsky. M.V.P.
was in part supported by RFBR grant 08-01-00054 and a grant of the RAS
Presidium \textquotedblright Fundamental Problems in Nonlinear
Dynamics\textquotedblright . He also thanks the Wroclaw University for the
hospitality. Authors are also grateful to Prof. Z. Popowicz (Wroc\l aw
University, Poland) for sending his report for the \ Symposium "Integrable
Systems: Zielona Gora - May 30-31, 2011" before its publication.

\end{document}